\documentclass[10pt, conference, letterpaper]{IEEEtran}
\IEEEoverridecommandlockouts
\usepackage{cite}
\usepackage{amsmath,amssymb,amsfonts}
\usepackage{algorithmic}
\usepackage{graphicx}
\usepackage{textcomp}
\usepackage{xcolor}
\definecolor{darkgreen}{rgb}{0.0, 0.39, 0.0}
\usepackage{enumitem}
\usepackage{url}
\usepackage{pifont}
\usepackage[linesnumbered,ruled,vlined]{algorithm2e}
\usepackage{booktabs}
\usepackage{multirow}
\usepackage{colortbl}
\usepackage{subfig}
\usepackage{hyperref}
\usepackage{cleveref}

\usepackage{mathtools}
\usepackage{amsthm}

\newtheorem{theorem}{Theorem}

\usepackage{subcaption}
\hypersetup{hidelinks}

\def\BibTeX{{\rm B\kern-.05em{\sc i\kern-.025em b}\kern-.08em
    T\kern-.1667em\lower.7ex\hbox{E}\kern-.125emX}}
\begin{document}
\bstctlcite{BSTcontrol}
\title{Inverting the Hidden: Unveiling Multimodal Privacy Leakage in Collaborative LVLM Inference
}

\author{
{\large\rmfamily
Shuaifan Jin$^{\dagger,\wr}$ \quad
Zhibo Wang$^{\dagger,\wr,\ast}$ \quad
Qiyuan Wang$^{\ddagger}$ \quad
Yiting Han$^{\dagger,\wr}$
}\\
{\large\rmfamily
Yajie Zhou$^{\dagger,\wr}$ \quad
Yuanfan Zhang$^{\dagger,\wr}$ \quad
Jiahui Hu$^{\S}$ \quad
Xiaoyi Pang$^{\natural}$
}\\[1mm]
{\normalsize
$^{\dagger}$The State Key Laboratory of Blockchain and Data Security, Zhejiang University, China
}\\
{\normalsize
$^{\wr}$College of Computer Science and Technology, Zhejiang University, China
}\\
{\normalsize
$^{\ddagger}$Department of Statistics, Texas A\&M University, USA
\quad
$^{\S}$Nanchang University, China
}\\
{\normalsize
$^{\natural}$Hong Kong University of Science and Technology, China
}\\[0.5mm]
{\small
\{shuaifanjin,zhibowang\}@zju.edu.cn
\quad
wqy213@tamu.edu
}\\
{\small
\{yitinghan,yajiezhou,zhangyuanfan\}@zju.edu.cn
\quad
jiahuihu@ncu.edu.cn
\quad
xypang@ust.hk
}
\thanks{$^\ast$Zhibo Wang is the corresponding author. Copyright may be transferred without notice, after which this version may no longer be accessible.}
}

\maketitle

\begin{abstract}

Collaborative inference deploys Large Vision-Language Models (LVLMs) by partitioning computation between edge devices and the cloud. While withholding raw inputs supposedly ensures privacy, transmitting intermediate hidden states exposes a critical attack surface. However, it remains unclear whether deep-layer LVLM hidden states retain recoverable private information, given that visual content has been projected into the language embedding space. To address this concern, we theoretically analyze LVLM hidden-state recoverability and show that, under regularity assumptions and a positive semantic--nuisance margin, privacy-relevant visual semantics remain identifiable and stably recoverable. Motivated by this analysis, we propose RASR, a novel coarse-to-fine multimodal reconstruction attack. RASR obtains initial image and text reconstructions through modality-specific inverse paths that follow their respective forward processing pipelines in reverse, and then uses hidden-state consistency to refine both reconstructions. Evaluations on Qwen3-VL-8B-Instruct and LLaVA-1.5-7B across five datasets demonstrate that RASR reduces image reconstruction MSE by \(\sim\)50\% compared to the strongest baselines, while achieving up to 99\% token accuracy for text recovery. These results show that privacy-sensitive visual and textual information can be recovered even from deep-layer LVLM hidden states, exposing the privacy risks of collaborative inference.

\end{abstract}


\section{Introduction}
Large Vision-Language Models (LVLMs) can jointly process visual and textual inputs, perceiving images and reasoning about them in natural language. This capability has made LVLMs widely used across a range of applications, including visual question answering, image captioning, and multimodal dialogue. However, the large parameter size and memory footprint of LVLMs make it difficult to run them directly on resource-constrained edge devices. A widely adopted solution is collaborative inference, in which the model is split into a front-end part that runs locally on the device and a back-end part that runs on a cloud server. The local device executes the first several layers and sends only the resulting intermediate hidden states to the cloud, which completes the remaining computation. Because the raw image and text never leave the local device, this approach is commonly believed to enable efficient inference while preserving data privacy.

In this paper, we question this presumed privacy guarantee. While raw inputs are withheld, prior research in unimodal contexts has demonstrated that transmitting intermediate representations does not inherently guarantee privacy. In language models, sensitive input text can be reconstructed from intermediate activations or gradients~\cite{chen2024unveiling,nikolaou2025language,zhao2025rep2text}. Similarly, in traditional vision models (e.g., CNNs and Vision Transformers), intermediate activations from shallow layers can be exploited to recover sensitive image content~\cite{dosovitskiy2016inverting,rathjens2024inverting,hu2024sparse}. These vulnerabilities raise an alarming security question for multimodal ecosystems: Since collaborative LVLM inference exposes intermediate hidden states derived from both visual and textual inputs, can an adversary reconstruct the user's original multimodal inputs from these hidden states? A recent survey on LVLM attacks also points out that the privacy landscape of LVLMs remains poorly understood~\cite{11127221}.

Answering this question is far from a trivial extension of existing unimodal inversion attacks. Although recent work~\cite{nikolaou2025language} has proved that text can be recovered from unimodal language-model hidden states by exploiting the discreteness of the text space, the recoverability of visual inputs from deep LVLM hidden states remains unclear, as visual reconstruction requires searching a high-dimensional continuous space without comparable discrete constraints. Moreover, the multimodal setting introduces further challenges. In LVLMs, the visual input is first encoded and projected into the language embedding space, and then processed together with the textual tokens through multiple Transformer layers.  As a result, the exposed visual hidden states have undergone both cross-modal projection and deep transformation, and each textual hidden state is conditioned on the visual prefix through causal self-attention. These challenges led us to answer two fundamental sub-questions for the original research question. First, \textit{do split-layer hidden states retain privacy-relevant visual semantics?} Second, if such information is retained, \textit{can an adversary reconstruct the original visual and textual inputs from these LVLM hidden states, and how effective is such reconstruction?}

To address these questions, we first analyze whether privacy-relevant visual semantics remain recoverable after cross-modal alignment and LLM propagation. By modeling local image variations using semantic and nuisance factors, we derive a sufficient condition for local semantic recovery. Specifically, under local regularity assumptions and a positive semantic--nuisance separation margin, privacy-relevant visual semantics remain locally identifiable and stably recoverable. Empirical measurements further show that the evaluated semantics remain decodable and that the measured margins remain positive across model depth, providing support for our analysis.

Motivated by this recoverability analysis, we propose RASR (Recoverability-Aware Symmetric Reconstruction), a coarse-to-fine multimodal reconstruction attack for recovering images and text instructions from intercepted hidden states. RASR employs two modality-specific reconstruction branches whose architectures mirror the corresponding forward pathways of the target LVLM. In the \textit{coarse} stage, the textual branch predicts all tokens in a single pass, and the visual branch progressively maps the visual hidden states back through the visual-embedding and visual-feature spaces to produce an initial image reconstruction. In the \textit{fine} stage, RASR uses hidden-state consistency to selectively correct low-confidence textual tokens through position-wise search and iteratively refine the image by minimizing the discrepancy between the intercepted and regenerated hidden states.

Our main contributions can be summarized as follows.

\begin{itemize}[leftmargin=*]
\item We provide the first theoretical analysis of hidden-state recoverability in collaborative LVLM inference. We derive an explicit sufficient condition under which privacy-relevant visual semantics remain locally identifiable and stably recoverable after cross-modal alignment and LLM propagation, and provide empirical evidence consistent with this condition on real LVLMs.

\item We propose RASR, a coarse-to-fine multimodal reconstruction attack that recovers images and text from an intercepted hidden-state sequence. RASR combines modality-specific symmetric reconstruction with hidden-state-consistency-guided refinement to progressively improve both visual and textual reconstructions.

\item Through extensive evaluations on leading LVLM architectures and diverse datasets, we demonstrate that RASR significantly outperforms existing baselines in reconstruction performance. Our findings bridge the gap between theoretical recoverability and practical attack efficacy, underscoring the urgent need for robust defenses in collaborative LVLM inference systems.
\end{itemize}
\section{Related Work}

\subsection{LVLMs and Collaborative Inference}
\subsubsection{\textbf{Large Vision-Language Models}}
Large Vision-Language Models bridge visual perception and language reasoning~\cite{10769058,Li_2025_CVPR}. Unlike earlier dual-encoder models such as CLIP~\cite{radford2021learning}, current LVLMs~\cite{liu2023visual,liu2024improved,zhang2024llavanextvideo,Qwen-VL,bai2025qwen25vltechnicalreport,bai2025qwen3vltechnicalreport} typically employ an LLM as the reasoning backbone and project visual features into its token space through an alignment module. However, their substantial computational and memory demands make end-to-end inference on resource-constrained devices challenging, motivating collaborative inference.

\subsubsection{\textbf{Collaborative Inference}}
In collaborative inference, model computation is partitioned across multiple devices. Neurosurgeon~\cite{kang2017neurosurgeon} studies layer-level partitioning of conventional DNNs between mobile devices and cloud servers. Recent systems extend this paradigm to large Transformer models: SplitLLM~\cite{mudvari2024splitllm} distributes LLM layers between clients and servers, Model-Distributed Inference~\cite{11154542} partitions LLMs across edge devices, and Petals~\cite{borzunov-etal-2023-petals} distributes Transformer blocks across geographically separated nodes. Although these systems exchange intermediate activations or hidden states, they primarily focus on deployment efficiency rather than the privacy risks of the transmitted representations.

\subsection{Privacy Inversion Attacks}
\subsubsection{\textbf{Vision Models}}
Prior studies show that intermediate CNN representations may retain sufficient information for reconstructing visual inputs~\cite{mahendran2015understanding,dosovitskiy2016inverting,zhang2020secret,298176}. Recent attacks extend image inversion to Transformer-based vision models using gradients, embedding optimization, sparse representations, or learned inverse mappings~\cite{hatamizadeh2022gradvit,kazemi2024we,hu2024sparse,rathjens2024inverting}. However, these studies focus on unimodal vision models and do not determine whether visual information remains recoverable after cross-modal alignment and propagation through an LLM backbone.

\subsubsection{\textbf{Language Models}}
Early language-model inversion attacks mainly recover specific attributes, patterns, or keywords~\cite{song2020information,pan2020privacy}. More recent methods reconstruct complete text from split-learning transmissions, model outputs, sentence embeddings, or hidden representations~\cite{chen2024unveiling,zhang2024extracting,morris2023text,qu2025prompt,zhao2025rep2text,nikolaou2025language}. Nevertheless, these methods are designed for unimodal language models and generally do not consider the visual context incorporated into LVLM hidden states.

Recent cross-modal attacks remain limited to recovering visual attributes or training data~\cite{xiu2025caprecover,nguyen2026visionlanguagemodelsleaklearn}. The reconstruction of both image and text inputs from intermediate hidden states in collaborative LVLM inference remains underexplored~\cite{11127221}. Our work addresses this gap by systematically investigating their recoverability from transmitted hidden states.
\section{Preliminary}

This section first describes how multimodal inputs are processed and how intermediate hidden states are transmitted from edge devices to the cloud, and then specifies the adversary's capabilities and objective.

\subsection{System Model}
\subsubsection{\textbf{LVLM Architecture}}
In this paper, we consider a widely adopted LVLM architecture that employs an LLM as the central backbone~\cite{liu2023visual,liu2024improved,zhang2024llavanextvideo,Qwen-VL,bai2025qwen25vltechnicalreport,bai2025qwen3vltechnicalreport}. It consists of a visual encoder $f_v$, also referred to as the vision tower, a cross-modal aligner $f_a$, and an LLM comprising a token embedding layer $f_e$, a Transformer backbone $f_{\mathrm{tr}}$, and an output head.

Formally, given an input image $\mathbf{I}$, the visual encoder extracts features $\mathbf{Z}=f_v(\mathbf{I})$, which the aligner maps to an ordered sequence of visual tokens $\mathbf{V}=f_a(\mathbf{Z})\in\mathbb{R}^{n\times d}$, where $n$ is the number of visual tokens and $d$ is the LLM hidden dimension. Let $\mathbf{T}=(t_1,\ldots,t_m)$ denote the tokenized text instruction, with embeddings $\mathbf{E}_T=f_e(\mathbf{T})\in\mathbb{R}^{m\times d}$. The concatenated multimodal sequence $[\mathbf{V};\mathbf{E}_T]$ is processed by the Transformer backbone, and the resulting contextualized representations are further processed by the output head to generate the response $\mathbf{Y}$ through autoregressive decoding. Special tokens, prompt-template tokens, and padding positions are omitted for clarity.

\subsubsection{\textbf{Collaborative Inference}}
In this paper, we consider a two-participant abstraction consisting of a front-end participant and a back-end participant, which also captures a single partition boundary in a multi-participant deployment~\cite{kang2017neurosurgeon,mudvari2024splitllm,11154542,borzunov-etal-2023-petals}. Specifically, the partition boundary is placed after the $s$-th layer of the Transformer backbone. The front-end participant hosts the visual encoder, the cross-modal aligner, the token embedding layer, and the first $s$ Transformer layers, while the back-end participant hosts the remaining Transformer layers and the output head. Let $F_{\mathrm{pre}}^{(s)}$ and $F_{\mathrm{post}}^{(s)}$ denote the Transformer sub-networks before and after the partition boundary, respectively. The Transformer backbone can therefore be decomposed as
$
f_{\mathrm{tr}}
=
F_{\mathrm{post}}^{(s)}
\circ
F_{\mathrm{pre}}^{(s)} .
$

The front-end participant first constructs the unified multimodal sequence and then computes the split-layer hidden states as
$
\mathbf{H}^{(s)}
=
\Phi_{\leq s}(\mathbf{I},\mathbf{T};\theta)
=
F_{\mathrm{pre}}^{(s)}
\bigl([\mathbf{V};\mathbf{E}_T]\bigr),
$
where $\Phi_{\leq s}$ denotes the complete front-end mapping up to the partition boundary, and $\theta$ contains the parameters of the visual encoder, cross-modal aligner, token embedding layer, and first $s$ Transformer layers. The transmitted sequence can be written as $\mathbf{H}^{(s)}=[\mathbf{H}_v^{(s)};\mathbf{H}_t^{(s)}]\in\mathbb{R}^{(n+m)\times d}$, where $\mathbf{H}_v^{(s)}\in\mathbb{R}^{n\times d}$ and $\mathbf{H}_t^{(s)}\in\mathbb{R}^{m\times d}$ denote the hidden states at the visual-token and textual-token positions, respectively. Here, the subscripts $v$ and $t$ refer only to the corresponding token positions in the original multimodal sequence, since both parts have already been contextualized by the Transformer layers.

The hidden states $\mathbf{H}^{(s)}$ are transmitted to the back-end participant, which processes them through the remaining Transformer layers and the output head to generate the final response $\mathbf{Y}$. In this work, we focus on the privacy risks posed by the transmission of $\mathbf{H}^{(s)}$.

\begin{figure}[!t]
    \centering
    \includegraphics[width=0.98\columnwidth]{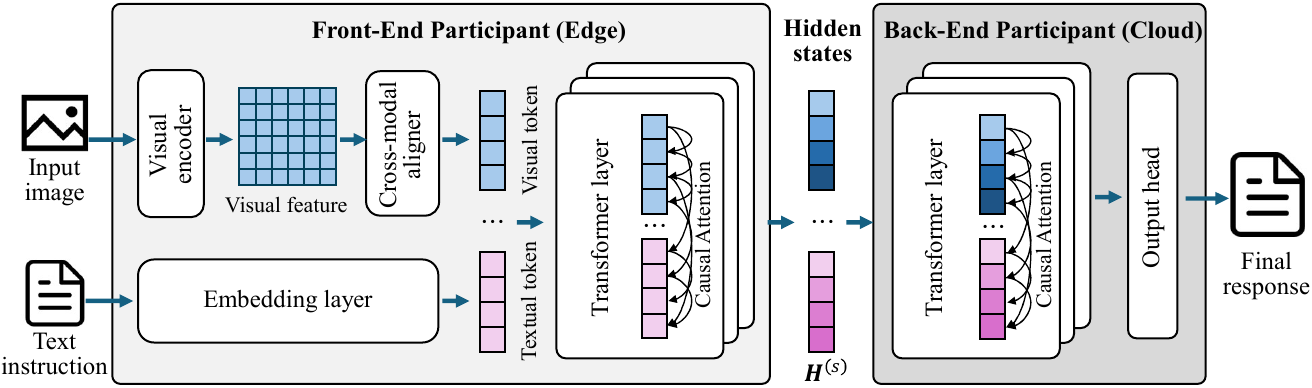}
    \caption{Collaborative inference pipeline of LVLMs.}
    \label{fig:system}
\end{figure}

\subsection{Threat Model}

\subsubsection{\textbf{Adversary Scenario}}
In this paper, we consider an \textit{honest-but-curious} back-end participant (e.g., a cloud server) that receives the transmitted hidden states during collaborative LVLM inference as the adversary. The adversary follows the prescribed inference protocol and does not modify the transmitted representations, interfere with model execution, or tamper with the final outputs. However, it attempts to infer private information from the intermediate representations exposed during collaborative inference. Specifically, the adversary can observe the transmitted hidden states $\mathbf{H}^{(s)}$, but cannot directly access the original image $\mathbf{I}$ or the text instruction $\mathbf{T}$.

\subsubsection{\textbf{Adversary Knowledge and Capability}}
Following prior work~\cite{chen2024unveiling,nikolaou2025language,mahendran2015understanding,qu2025prompt}, we assume that the adversary knows the target LVLM architecture and collaborative inference configuration, including the partition boundary, hidden-state dimensions, and sequence structure. The adversary also has access to the parameters of the front-end mapping $\Phi_{\leq s}$, enabling it to reproduce the forward computation and propagate gradients up to the partition boundary. Such access is feasible when the deployed LVLM weights are publicly available or the corresponding model checkpoint is accessible to participating parties. The adversary may use public auxiliary image-text data, but cannot access the victim's original inputs or the front-end participant's private runtime environment.

\subsubsection{\textbf{Adversary Objective}}
The adversary's primary objective is to compromise user privacy by recovering sensitive visual and textual information from the transmitted hidden states. The adversary does not aim to disrupt the normal collaborative inference process, but instead observes the hidden states $\mathbf{H}^{(s)}$ and reconstructs the victim's private inputs:
$
(\hat{\mathbf{I}},\hat{\mathbf{T}})
=
R_{\psi}\bigl(\mathbf{H}^{(s)}\bigr),
\label{eq:adversary_objective}
$
where $R_{\psi}$ denotes the adversary's reconstruction mapping parameterized by $\psi$, and $\hat{\mathbf{I}}$ and $\hat{\mathbf{T}}$ denote the reconstructed image and text instruction, respectively.
\section{Recoverability of LVLM Hidden States}
\label{sec:recoverability}

This section gives a sufficient condition for locally recovering privacy-relevant visual semantics from exposed LVLM hidden states without requiring global or exact pixel-level invertibility. The condition is quantified by a quotient margin that separates semantic responses from nuisance variations and is estimated using Jacobian--vector products in Section~\ref{sec:rq1-recoverability}.

\subsection{Quotient-Transverse Semantic Recoverability}

Fix a text instruction $\mathbf{T}$ and a split layer $s$. In a neighborhood of a natural image, let
\begin{equation}
    \mathbf{I}=G(c,u), \quad
    h_s(c,u)=\operatorname{vec}\!\left(\Phi_{\leq s}(G(c,u),\mathbf{T};\theta)\right).
    \label{eq:local-hidden-map}
\end{equation}
where $G$ is a local parameterization of natural images, $c\in\mathbb{R}^{d_c}$ denotes selected perceptually relevant semantic coordinates, $u\in\mathbb{R}^{d_u}$ denotes nuisance factors not specified by $c$, and $h_s(c,u)\in\mathbb{R}^{d_h}$ denotes the vectorized layer-$s$ hidden state. Define
\begin{equation}
\begin{aligned}
    J_c &= \frac{\partial h_s}{\partial c},
    & J_u &= \frac{\partial h_s}{\partial u},\\
    \Pi_u^\perp &= I_{d_h}-J_uJ_u^\dagger,
    & S_s &= \Pi_u^\perp J_c,
\end{aligned}
\label{eq:quotient-semantic-response}
\end{equation}
where $I_{d_h}$ is the $d_h$-dimensional identity matrix, ${}^{\dagger}$ denotes the Moore--Penrose pseudoinverse, and $\Pi_u^\perp$ is the orthogonal projector onto the complement of the nuisance-response space $\operatorname{Im}(J_u)$. Thus, $S_s$ measures the hidden-state response to semantic variations after removing components explainable by nuisance variations.

At a reference point $x_0=(c_0,u_0)$, assume that:
\begin{enumerate}[
    label=(C\arabic*),
    leftmargin=*,
    itemsep=1pt,
    topsep=2pt
]
    \item $h_s$ is twice continuously differentiable on an open neighborhood of $x_0$;
    \item $J_u$ has constant rank $r$ on a neighborhood of $x_0$;
    \item the quotient semantic margin
    $\gamma_0:=\sigma_{\min}(S_s(c_0,u_0))$ is positive.
\end{enumerate}
Condition~(C3) requires every nonzero semantic direction to induce a hidden-state variation that cannot be locally canceled by nuisance-induced variations. This condition depends only on the local geometry of the forward hidden-state mapping and does not presuppose the existence of an inverse decoder.

\begin{theorem}[Local semantic recovery from quotient transversality]
\label{thm:quotient-recoverability}
Under Conditions~(C1)--(C3), there exist a neighborhood $W_0$ of $x_0$, a local hidden-state manifold $\mathcal{M}_0=h_s(W_0)$, and a continuously differentiable local semantic decoder $\phi_s:\mathcal{M}_0\rightarrow\mathbb{R}^{d_c}$ such that
\begin{equation}
    \phi_s(h_s(c,u))=c,
    \qquad (c,u)\in W_0.
    \label{eq:local-semantic-decoder}
\end{equation}
Moreover, $W_0$ can be chosen such that
\begin{equation}
    \underline{\gamma}
    :=
    \inf_{(c,u)\in W_0}
    \sigma_{\min}\!\left(S_s(c,u)\right)
    \geq \frac{\gamma_0}{2}>0.
    \label{eq:neighborhood-quotient-margin}
\end{equation}
For any $(c,u)\in W_0$ and any tangent hidden-state perturbation
\begin{equation}
    e=J_c(c,u)\Delta c+J_u(c,u)\Delta u
    \in T_{h_s(c,u)}\mathcal{M}_0,
\end{equation}
the semantic decoder satisfies
\begin{equation}
    \left\|D\phi_s(h_s(c,u))e\right\|_2
    =\|\Delta c\|_2
    \leq\underline{\gamma}^{-1}\|e\|_2.
    \label{eq:differential-stability}
\end{equation}
Consequently, for $h,h'\in\mathcal{M}_0$ with $h'\rightarrow h$,
\begin{equation}
    \|\phi_s(h')-\phi_s(h)\|_2
    \leq
    \underline{\gamma}^{-1}\|h'-h\|_2
    +o(\|h'-h\|_2).
    \label{eq:finite-local-stability}
\end{equation}
\end{theorem}

\begin{proof}
By Condition~(C2), $J_u^\dagger$ and hence $\Pi_u^\perp$ vary continuously in a sufficiently small neighborhood of $x_0$. Condition~(C3) and the continuity of singular values then allow the neighborhood to be chosen such that $\sigma_{\min}(S_s)\geq\gamma_0/2$, which gives~\eqref{eq:neighborhood-quotient-margin}.

For $(\Delta c,\Delta u)\in\ker Dh_s(c,u)$,
\[
    J_c\Delta c+J_u\Delta u=0
    \quad\Longrightarrow\quad
    S_s\Delta c=0
    \quad\Longrightarrow\quad
    \Delta c=0,
\]
where the first implication follows by applying $\Pi_u^\perp$, and the second follows because $S_s$ has full column rank. Therefore, with $\pi_c(c,u):=c$, we have
\begin{equation}
    \ker Dh_s(c,u)
    \subseteq
    \ker D\pi_c(c,u).
    \label{eq:kernel-inclusion}
\end{equation}
Since $\Pi_u^\perp J_c$ lies in the orthogonal complement of $\operatorname{Im}(J_u)$,
\[
\begin{aligned}
    \operatorname{rank}Dh_s
    &=\operatorname{rank}[J_c\;J_u]\\
    &=\operatorname{rank}(J_u)
      +\operatorname{rank}(\Pi_u^\perp J_c)\\
    &=r+d_c.
\end{aligned}
\]
Hence, after shrinking $W_0$ if necessary, $h_s$ has constant rank $r+d_c$. The constant-rank theorem then allows $W_0$ to be chosen such that $\mathcal{M}_0=h_s(W_0)$ is an embedded submanifold and each fiber of $h_s|_{W_0}$ is connected. Together with~\eqref{eq:kernel-inclusion}, this implies that $\pi_c$ is constant along each such fiber. The local factorization theorem therefore yields a continuously differentiable map $\phi_s$ satisfying $\phi_s\circ h_s=\pi_c$, which proves~\eqref{eq:local-semantic-decoder}.

Differentiating $\phi_s\circ h_s=\pi_c$ gives
\begin{equation}
    D\phi_sJ_c=I_{d_c},
    \qquad
    D\phi_sJ_u=0.
\end{equation}
Thus, for $e=J_c\Delta c+J_u\Delta u$, we have $D\phi_s e=\Delta c$. Moreover,
\begin{equation}
    \|e\|_2\geq\|\Pi_u^\perp e\|_2=\|S_s\Delta c\|_2\geq\underline{\gamma}\|\Delta c\|_2,
\end{equation}
which proves~\eqref{eq:differential-stability}.

Since $\mathcal{M}_0$ is a local $C^2$ embedded manifold, for $h'$ sufficiently close to $h$, there exists a $C^1$ curve in $\mathcal{M}_0$ joining $h$ and $h'$ whose length is $\|h'-h\|_2+o(\|h'-h\|_2)$ as $h'\rightarrow h$. Integrating the differential bound along this curve yields~\eqref{eq:finite-local-stability}.
\end{proof}

\subsection{Scope of the Result}
Theorem~\ref{thm:quotient-recoverability} establishes a locally stable semantic left inverse for the selected semantic coordinates $c$, rather than a full inverse for $(c,u)$. This existence result motivates the learned reconstruction model introduced in Section~\ref{sec:method}, while recovery beyond the selected semantic coordinates is evaluated empirically. Section~\ref{sec:rq1-recoverability} examines finite-dimensional signatures of the sufficient condition, including semantic readability, non-degenerate semantic responses, and the quotient margin after removing measured nuisance directions. Because these quantities are estimated using selected semantic and nuisance directions at finitely many samples and split layers, the experiments provide evidence consistent with local recoverability rather than establishing a uniform guarantee.
\begin{figure*}[!t]
    \centering

    % Left: pipeline
    \begin{minipage}[c]{0.74\textwidth}
        \centering
        \subfloat[Overall pipeline of RASR with symmetric reconstruction and consistency-guided refinement.\label{fig:pipe}]{
            \includegraphics[width=\linewidth]
            {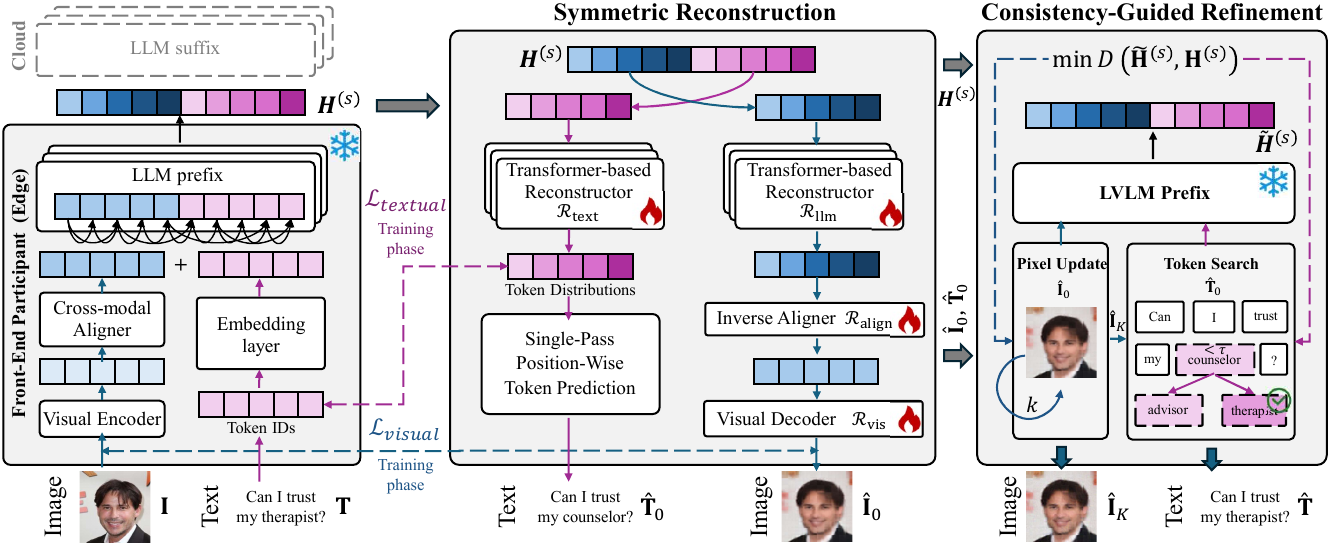}
        }
    \end{minipage}
    % \hfill
    \hspace{1em}
    % Right: feature visualizations
    \begin{minipage}[c]{0.158\textwidth}
        \centering
        \subfloat[Feature space (LLaVA).\label{fig:llava_feature_space_visualization}]{
            \includegraphics[width=1\linewidth]
            {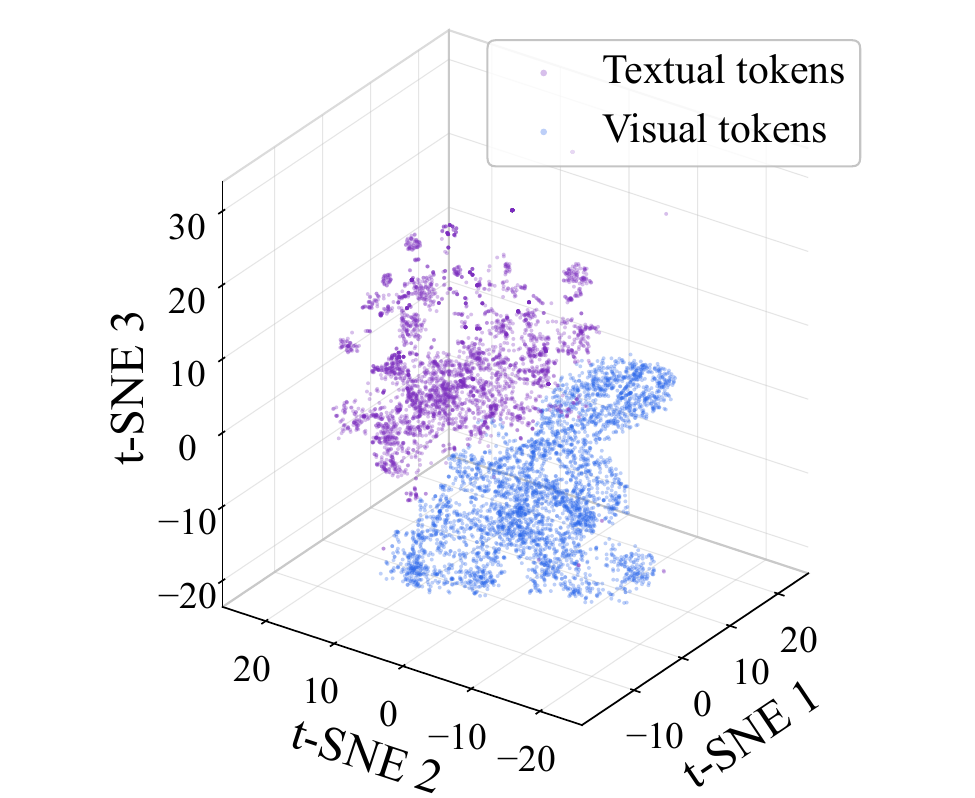}
        }

        \subfloat[Feature space (Qwen).\label{fig:qwen_feature_space_visualization}]{
            \includegraphics[width=1\linewidth]
            {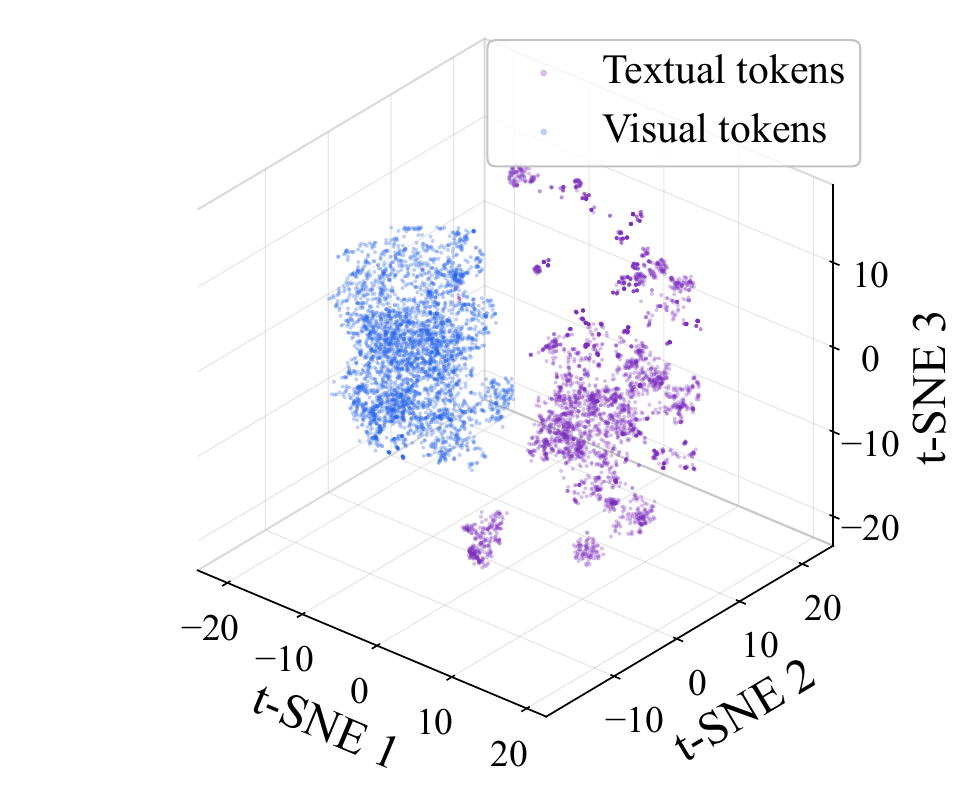}
        }
    \end{minipage}

    \caption{Overview of the proposed reconstruction pipeline and the
    feature-space distributions of visual and textual hidden states.}
    \label{fig:pipeline_and_feature_space}
\end{figure*}

\section{Method}
\label{sec:method}

This section presents our inverse attack for revealing privacy risks in large vision-language models (LVLMs). After introducing the overall framework in \cref{sec:Overview}, we detail the symmetric reconstruction architecture in \cref{sec:symmetric_reconstruction} and the consistency-guided refinement strategy in \cref{sec:consistency_refinement}.

\subsection{Overview}\label{sec:Overview}
Building on the preceding recoverability analysis, we design an inverse attack to exploit the recoverable information preserved in split LVLM hidden states. The analysis shows that the split-layer mapping does not uniformly collapse visual information. Instead, privacy-relevant visual variations can remain distinguishable in the hidden-state space, suggesting that split LVLM hidden states may preserve recoverable visual semantics even though they are no longer raw inputs. Motivated by this observation, we construct modality-specific inverse mappings and further refine their outputs through hidden-state consistency to recover privacy-relevant input information.

Fig.~\ref{fig:pipe} illustrates the overall pipeline of our inverse attack. It follows a coarse-to-fine strategy comprising two modules: Symmetric Reconstruction and Consistency-Guided Refinement. Given the intercepted hidden states, Symmetric Reconstruction first separates the visual and textual components and obtains initial image and text reconstructions through modality-specific inverse paths that approximately reverse the corresponding forward stages. Consistency-Guided Refinement then updates these reconstructions to reduce the discrepancy between their re-encoded hidden states and the intercepted target hidden states.

\subsection{Symmetric Reconstruction Architecture}
\label{sec:symmetric_reconstruction}

In collaborative LVLM inference, the visual input is encoded by the visual encoder, projected into the language embedding space by the cross-modal aligner, concatenated with textual tokens, and subsequently processed by the LLM backbone through causal self-attention. Consequently, the representations transmitted across the partition boundary are contextualized hidden states in a shared LLM hidden space rather than isolated modality-specific features. Although cross-modal interaction progressively transforms these representations, the visual-prefix layout and the positional correspondence of the original multimodal sequence remain unchanged.

To explore whether modality-dependent structure is preserved after such propagation, we visualize the hidden states at visual and textual token positions using t-SNE. As shown in \cref{fig:llava_feature_space_visualization,fig:qwen_feature_space_visualization}, the two modalities still form distinguishable distributions, suggesting that LLM contextualization does not completely eliminate their modality-dependent organization. Given the intercepted sequence $\mathbf{H}^{(s)}=(h_1,\ldots,h_L)$ and the known visual-prefix length $b$, we separate it into $\mathbf{H}_v=(h_1,\ldots,h_b)$ and $\mathbf{H}_t=(h_{b+1},\ldots,h_L)$. After separation, the two modalities require different reconstruction strategies because visual hidden states are produced from images through multiple successive continuous transformations, whereas textual hidden states retain a position-wise correspondence with discrete input tokens. Thus, we process $\mathbf{H}_v$ and $\mathbf{H}_t$ using dedicated visual and textual branches that approximately reverse the major transformations along their respective forward paths.

Following this symmetric reconstruction design, we begin with visual reconstruction. Rather than mapping $\mathbf{H}_v$ directly to raw pixels with a single reconstruction network, we decompose the inverse mapping according to the major stages of the forward visual path. In the forward direction, the image is mapped to visual features by the visual encoder, projected into aligned visual tokens by the cross-modal aligner, and subsequently transformed into the intercepted visual hidden states by the LLM layers preceding the partition. Accordingly, we construct a hierarchical visual branch that approximately reverses this transformation sequence:
\begin{equation}
\mathbf{I}
\rightarrow
\mathbf{Z}
\rightarrow
\mathbf{V}
\rightarrow
\mathbf{H}_v,
\quad
\mathbf{H}_v
\xrightarrow{\mathcal{R}_{\mathrm{llm}}}
\hat{\mathbf{V}}
\xrightarrow{\mathcal{R}_{\mathrm{align}}}
\hat{\mathbf{Z}}
\xrightarrow{\mathcal{R}_{\mathrm{vis}}}
\hat{\mathbf{I}}.
\label{eq:hierarchical_visual_reconstruction}
\end{equation}
Specifically, since $\mathbf{H}_v$ and $\mathbf{V}$ are both ordered token sequences with position-wise correspondence, a Transformer is naturally compatible with the required sequence-to-sequence mapping while capturing cross-token dependencies. We therefore implement $\mathcal{R}_{\mathrm{llm}}$ as a Transformer-based reconstructor that maps $\mathbf{H}_v$ back to the aligned visual tokens $\hat{\mathbf{V}}$. $\mathcal{R}_{\mathrm{align}}$ then projects $\hat{\mathbf{V}}$ into the visual-feature space to obtain $\hat{\mathbf{Z}}$. Finally, $\mathcal{R}_{\mathrm{vis}}$ restores the spatial patch layout and progressively decodes $\hat{\mathbf{Z}}$ into the reconstructed image $\hat{\mathbf{I}}$.

With the hierarchical visual branch established, we optimize its three stages jointly rather than imposing separate reconstruction objectives on their intermediate outputs. Specifically, although structurally aligned with the corresponding forward-path representations, $\hat{\mathbf{V}}$ and $\hat{\mathbf{Z}}$ serve only as latent transition states rather than explicit reconstruction targets. We do not supervise them directly because enforcing point-wise agreement with the original activations could unnecessarily restrict the latent solution space and hinder final image reconstruction. Accordingly, we define the optimization objective solely in the image space:
\begin{equation}
\mathcal{L}_{\mathrm{visual}}
=
\frac{1}{|\Omega|}
\left\|
\hat{\mathbf{I}}-\mathbf{I}
\right\|_{2}^{2},
\label{eq:visual_reconstruction_loss}
\end{equation}
where $\mathbf{I}$ denotes the original image and $|\Omega|$ is the number of image elements.

For textual reconstruction, the intercepted textual hidden states preserve both their position-wise correspondence with the input tokens and the causal dependency structure induced by the decoder-only Transformer backbone $\mathcal R_{\mathrm{text}}$. Therefore, we employ a structurally matched decoder-only Transformer parameterized by $\psi_t$ to recover the aligned token sequence. After re-indexing the separated textual positions, let $\mathbf{H}_t=(h_1^t,\ldots,h_m^t)$ and $\mathbf{T}=(t_1,\ldots,t_m)$ denote the textual hidden states and the corresponding ground-truth tokens, respectively. Under a causal attention mask, the output at position $i$ depends only on the hidden-state prefix $\mathbf{H}_t^{\leq i}=(h_1^t,\ldots,h_i^t)$ and predicts the aligned token $t_i$. Accordingly, the optimization objective is the position-wise cross-entropy loss:
\begin{equation}
\mathcal{L}_{\mathrm{textual}}
=
-\frac{1}{m}
\sum_{i=1}^{m}
\log p_{\psi_t}^{(i)}
\left(
t_i\mid\mathbf{H}_t^{\leq i}
\right),
\label{eq:text_reconstruction_loss}
\end{equation}
where $p_{\psi_t}^{(i)}(t_i\mid\mathbf{H}_t^{\leq i})$ is the probability of $t_i$ at position $i$. Since predicted tokens are not fed back into the reconstructor, causal masking does not require sequential decoding, and all positions are recovered in a single forward pass.

Overall, the visual and textual reconstruction branches are optimized using $\mathcal{L}_{\mathrm{visual}}$ and $\mathcal{L}_{\mathrm{textual}}$, respectively, thereby learning modality-specific inverse mappings from the intercepted hidden states to the image and text spaces. Their outputs, denoted by $(\hat{\mathbf{I}}_0,\hat{\mathbf{T}}_0)$, serve as the initial reconstructions for the subsequent consistency-guided refinement stage.

\subsection{Consistency-Guided Refinement}
\label{sec:consistency_refinement}

Although the symmetric reconstruction stage provides initial reconstructions, its learned inverse mappings may not fully exploit the information preserved in the intercepted hidden states. Therefore, we refine these reconstructions by matching their regenerated hidden states to the intercepted ones:
\begin{equation}
\min_{\mathbf{I}',\mathbf{T}'}
D\left(
\widetilde{\mathbf{H}}^{(s)}(\mathbf{I}',\mathbf{T}'),
\mathbf{H}^{(s)}
\right),
\label{eq:general-refinement}
\end{equation}
where $\mathbf{I}'$ and $\mathbf{T}'$ denote the candidate reconstructions, $\widetilde{\mathbf{H}}^{(s)}(\mathbf{I}',\mathbf{T}')=\Phi_{\leq s}(\mathbf{I}',\mathbf{T}';\theta)$ denotes their hidden states regenerated by the frozen front-end mapping, and $D(\cdot,\cdot)$ measures the hidden-state discrepancy.

For visual refinement, we fix $\hat{\mathbf{T}}_0$ and start from $\hat{\mathbf{I}}_0$. At iteration $k=0,\ldots,K-1$, we regenerate $\widetilde{\mathbf{H}}_k^{(s)}=\widetilde{\mathbf{H}}^{(s)}(\hat{\mathbf{I}}_k,\hat{\mathbf{T}}_0)$ and update the image by
\begin{equation}
    \hat{\mathbf{I}}_{k+1}
    =
    \Pi_{\mathcal{N}(\hat{\mathbf{I}}_0)}
    \left[
    \hat{\mathbf{I}}_k
    -
    \eta_k\nabla_{\hat{\mathbf{I}}_k}
    D\left(
    \widetilde{\mathbf{H}}_k^{(s)},
    \mathbf{H}^{(s)}
    \right)
    \right],
    \label{eq:image-refinement}
\end{equation}
where $\eta_k$ is the step size and $\Pi_{\mathcal{N}(\hat{\mathbf{I}}_0)}$ denotes projection onto the neighborhood $\mathcal{N}(\hat{\mathbf{I}}_0)$. The resulting $\hat{\mathbf{I}}_K$ is used as the visual context for textual refinement.

For textual refinement, we fix $\hat{\mathbf{I}}_K$, initialize $\hat{\mathbf{T}}=\hat{\mathbf{T}}_0$, and refine the low-confidence positions $\mathcal{P}_{\mathrm{opt}}=\{i\mid q_i<\tau\}$ from left to right, where $q_i$ is the highest token confidence at position $i$ and $\tau$ is the threshold. For each position $i$, we initialize a continuous variable $\mathbf{e}_i$ with its current token embedding and optimize it to minimize the hidden-state discrepancy, yielding $\mathbf{e}_i^\ast$. The candidate set $\mathcal{C}_i$ comprises the eight vocabulary tokens whose embeddings are nearest to $\mathbf{e}_i^\ast$ in Euclidean distance. Let $\hat{\mathbf{T}}^{(i\leftarrow v)}$ denote the current reconstruction with its $i$th token replaced by $v$, and $\widetilde{h}_i^t(v)$ the textual hidden state at position $i$ regenerated from $\bigl(\hat{\mathbf{I}}_K,\hat{\mathbf{T}}^{(i\leftarrow v)}\bigr)$. We select
\begin{equation}
    \hat{t}_i
    =
    \arg\min_{v\in\mathcal{C}_i}
    D\left(
    \widetilde{h}_i^t(v),
    h_i^t
    \right),
    \label{eq:text-refinement}
\end{equation}
where $h_i^t$ is the intercepted textual hidden state at position $i$. The selected token updates $\hat{\mathbf{T}}$ before the next position is processed, while positions outside $\mathcal{P}_{\mathrm{opt}}$ remain unchanged.

\section{Experiments}
\label{sec:experiments}
To evaluate the recoverability of hidden states and the effectiveness of our reconstruction attack, we conduct comprehensive experiments to answer the following research questions:

\begin{itemize}[leftmargin=*]
    \item \textbf{RQ1: Do split-layer hidden states retain privacy-relevant visual semantics?}
    \item \textbf{RQ2: How effectively can the proposed attack reconstruct visual and textual inputs from hidden states?}
\end{itemize}

\subsection{Experimental Setup}
\subsubsection{\textbf{Datasets and Pre-processing}}
For visual reconstruction, we use VQAv2~\cite{goyal2017making}, CelebA~\cite{liu2015deep}, and Oxford Flowers-102~\cite{nilsback2008automated}, with 88752/1000, 36468/1000, and 6551/1000 training/evaluation samples, respectively. VQAv2 and CelebA follow their official training/validation and training/test splits, while Flowers-102 uses a class-stratified split. VQAv2 uses its associated questions, whereas CelebA and Flowers-102 use annotation-derived instructions and are evaluated only on image reconstruction. For textual reconstruction, we use 88752/120 VQAv2 samples, 21962/120 FEVER~\cite{thorne2018fever} samples, and 20117/120 MS MARCO~\cite{nguyen2016ms} samples. Since FEVER and MS MARCO contain no images, their texts are paired with images from the corresponding VQAv2 training or evaluation pool. The selected VQAv2 training and evaluation subsets are disjoint at the image-ID level. All subsets and pairings are fixed using random seed 1024.

Images are resized to $336\times336$ and normalized using the default parameters of each LVLM. Texts are formatted and tokenized using the corresponding chat template and tokenizer.

\subsubsection{\textbf{Models and Implementation Details}}
We evaluate Qwen3-VL-8B-Instruct (Qwen) and LLaVA-1.5-7B (LLaVA), with 36 and 32 Transformer layers, respectively. The models are partitioned after layers 23 and 20, roughly two-thirds of each backbone. All target-model parameters remain frozen throughout training and refinement.

For each target model, both reconstruction branches use separately parameterized 1-layer causal Transformer blocks trained from scratch. The blocks follow the designs of Qwen3-1.7B and Llama-3.2-1B for Qwen and LLaVA, respectively, with a hidden size of 4,096. The visual branch further employs a 2-layer MLP projector, a model-specific inverse vision module, and a u-net style convolutional reconstruction backbone. The corresponding visual feature dimensions are 1,152 and 1,024, with patch sizes of 16$\times$16 and 14$\times$14, respectively. The textual branch projects the mapped hidden states onto the corresponding target-model vocabulary for position-wise prediction. The visual and textual reconstructors are trained for 3 and 2 epochs with batch sizes of 3 and 1, respectively.

All reconstructors use AdamW with a learning rate of 3$\times$10$^{-5}$, a linear schedule, and 100 warm-up steps. For consistency-guided refinement, the discrepancy $D(\cdot,\cdot)$ is defined as mean squared error. Image refinement runs for 100 iterations ($K$=100) using AdamW, with learning rates of 4.5$\times$10$^{-3}$ and 1.4$\times$10$^{-2}$ for Qwen and LLaVA, respectively. After each update, the image is projected to satisfy a normalized $L_2$ (RMS) radius of 0.01 from the initial reconstruction and a maximum per-pixel deviation of 0.05, with pixel values clipped to $[0,1]$. Textual refinement uses a confidence threshold of $\tau$=0.85 and a candidate-set size of 8. All experiments are conducted on NVIDIA RTX 5880 Ada Generation GPUs. End-to-end latency is measured with a batch size of 1.

\subsubsection{\textbf{Methods for Comparison}}
We compare our method with five representative inversion baselines covering optimization-based and learned reconstruction. For images, CLIPInversion~\cite{kazemi2024we} performs pixel optimization through representation matching, SMI~\cite{hu2024sparse} uses attention-guided sparse inversion, and INVERSE-TVM~\cite{rathjens2024inverting} learns module-wise inverse mappings. For text, PIA~\cite{qu2025prompt} combines embedding optimization with adaptive discretization, while SIPIT~\cite{nikolaou2025language} sequentially recovers tokens through continuous optimization and vocabulary projection. We adapt each baseline to the separated visual or textual hidden states at the target partition layer while preserving its original inversion mechanism. All methods are evaluated using identical data splits, intercepted hidden states, and target-model access, with the same input resolution used for visual reconstruction.

\subsubsection{\textbf{Evaluation Metrics}}
For visual reconstruction, we report MSE, PSNR, and SSIM. MSE measures pixel-wise distortion, PSNR measures reconstruction fidelity, and SSIM evaluates structural similarity. Lower MSE and higher PSNR and SSIM indicate better performance. For textual reconstruction, we report Token Accuracy (Token Acc.), Exact Match Rate (EMR), and BERTScore F1 (BERT-F1), which measure position-wise token recovery, exact sequence recovery, and semantic similarity, respectively. Higher values indicate better performance. Notably, only tokens corresponding to the original textual content contribute to the textual evaluation.
\subsection{Empirical Validation of Hidden-State Recoverability (RQ1)}
\label{sec:rq1-recoverability}

Following Section~\ref{sec:recoverability}, we examine whether privacy-relevant visual semantics remain (i) linearly decodable from hidden state, (ii) non-degenerate along semantic directions, and (iii) transverse to measured nuisance variations. We evaluate Qwen at layers 0, 6, 11, 17, 23, 29, and 35, and LLaVA at layers 0, 5, 10, 15, 20, 25, and 31.

\begin{table}[t]
\centering
\caption{Macro-average balanced accuracy of ten privacy-relevant CelebA attribute probes across model depths.}
\label{tab:semantic-probes}
\scriptsize
\setlength{\tabcolsep}{2.8pt}
\renewcommand{\arraystretch}{0.95}
\begin{tabular}{@{}llccccccc@{}}
\toprule
Model & & \multicolumn{7}{c}{Evaluated layers} \\
\cmidrule(lr){3-9}
\multirow{2}{*}{Qwen}
& Layer
& 0 & 6 & 11 & 17 & 23 & 29 & 35 \\
& Bal. Acc.
& 0.8213 & 0.8077 & 0.8039 & 0.8015
& 0.7971 & 0.8004 & 0.7708 \\
\midrule
\multirow{2}{*}{LLaVA}
& Layer
& 0 & 5 & 10 & 15 & 20 & 25 & 31 \\
& Bal. Acc.
& 0.8267 & 0.8067 & 0.8001 & 0.7853
& 0.7904 & 0.7798 & 0.7866 \\
\bottomrule
\end{tabular}
\end{table}

\begin{figure}[t]
\centering
\includegraphics[width=0.235\columnwidth]
{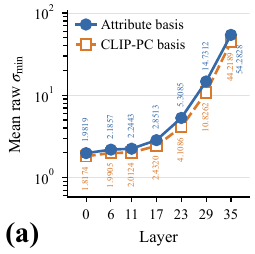}
\hfill
\includegraphics[width=0.235\columnwidth]
{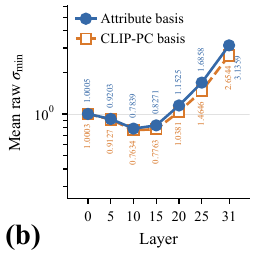}
\hfill
\includegraphics[width=0.235\columnwidth]
{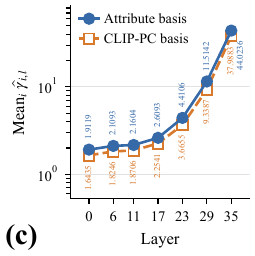}
\hfill
\includegraphics[width=0.235\columnwidth]
{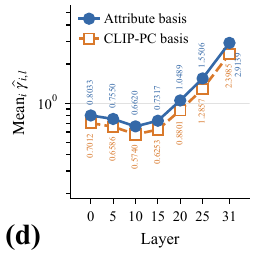}
\caption{Exact-JVP results over 100 evaluation inputs per layer and semantic basis. Panels (a)--(b) report the mean $\sigma_{\min}(\widehat J_{c,i,l})$ for Qwen and LLaVA, while panels (c)--(d) report the corresponding mean quotient margin $\widehat\gamma_{i,l}=\sigma_{\min}(\widehat S_{i,l})$. Error bars denote 95\% input-bootstrap confidence intervals.}
\label{fig:recoverability-validation}
\end{figure}

\subsubsection{\textbf{Semantic Decodability}}
\label{sec:e1-semantic-probes}

We mean-pool visual-token hidden states and fit independent ridge classifiers for ten CelebA attributes: Male, Smiling, Eyeglasses, Young, Black Hair, Blond Hair, Brown Hair, Mustache, No Beard, and Wearing Hat. Using disjoint sets of 5,000 training and 2,000 test images, the macro-average balanced accuracy ranges from 0.7708 to 0.8213 for Qwen and from 0.7798 to 0.8267 for LLaVA (Table~\ref{tab:semantic-probes}). Every attribute-layer probe at every evaluated layer achieves above 0.5 balanced accuracy, with minimum values of 0.5935 for Qwen and 0.5944 for LLaVA. The results indicate that privacy-relevant attributes remain linearly decodable from hidden states across different model depths.

\subsubsection{\textbf{Semantic Non-Degeneracy}}
\label{sec:e2-e4-exact-jvp}
Let $z_i$ be the vectorized projected visual-token representation of input $i$, and let $F_l(z,\mathbf{T}_0)$ denote the vectorized layer-$l$ hidden state under $\mathbf{T}_0=\text{``Describe the image in detail.''}$. For a model-specific semantic basis $V_c$, we compute $\widehat J_{c,i,l}=D_zF_l(z_i,\mathbf{T}_0)V_c$ using exact JVPs via automatic differentiation. This quantity serves as a token-space proxy for $J_c$ in Eq.~\eqref{eq:quotient-semantic-response}, with $V_c$ representing the selected semantic directions in the projected space. 
%rather than through the local image parameterization $G(c,u)$.

We use 192 images to construct two orthonormal semantic bases and a disjoint set of 100 images for evaluation. The first is a 9-dimensional CelebA basis formed from the mean differences between positive and negative examples of each attribute. We exclude Mustache because the basis set contains only one positive example. The second is a 16-dimensional CLIP-PC basis, where we map the leading principal components of normalized CLIP image embeddings into the centered projected-token space. Rank and $\sigma_{\min}$ are computed separately for each input before aggregation, with numerical rank determined using the matrix-specific tolerance $\max\{10^{-10},10^{-6}\sigma_{\max}\}$. All 2,800 semantic-response matrices $\widehat J_{c,i,l}$ are full column rank. As shown in Fig.~\ref{fig:recoverability-validation}(a)--(b), the mean $\sigma_{\min}$ ranges are 1.9819--54.2828/1.8174--44.2189 for Qwen and 0.7839--3.1359/0.7634--2.6544 for LLaVA under the CelebA/CLIP-PC bases, respectively. These results indicate that the hidden state responds to every tested semantic variation along both semantic subspaces.

\subsubsection{\textbf{Semantic--Nuisance Transversality}}
\label{sec:e3-transversality}

From the same 192 basis images, we construct 960 unit-normalized projected-token secants using paired brightness (1.15/0.85), contrast (1.15/0.85), saturation (1.20/0.80), hue ($+0.04/-0.04$), and sharpness (1.30/0.70) transformations. Their leading 16 right singular vectors define a model-specific nuisance basis $V_u$ shared across inputs and layers. In particular, $V_c$ and $V_u$ are orthonormalized separately. For each input $i$ and layer $l$, we use automatic differentiation to compute $\widehat J_{u,i,l}=D_zF_l(z_i,\mathbf{T}_0)V_u$, and then obtain $\widehat{\Pi}_{u,i,l}^{\perp} =I_{d_h}-\widehat J_{u,i,l}\widehat J_{u,i,l}^{\dagger}$, $\widehat S_{i,l} =\widehat{\Pi}_{u,i,l}^{\perp}\widehat J_{c,i,l}$, and $\widehat\gamma_{i,l} =\sigma_{\min}(\widehat S_{i,l})$. For each layer and semantic basis, 
we average $\widehat\gamma_{i,l}$ across evaluation inputs separately.

Across all 2,800 nuisance-projected semantic-response matrices $\widehat S_{i,l}$, every matrix remains full column rank, with $\widehat\gamma_{i,l}>0$. All 700 LLaVA nuisance-response matrices have rank 16. For Qwen, 693 matrices have rank 16 and the remaining seven have ranks between 11 and 15. Despite this lower nuisance rank, the corresponding projected semantic-response matrices remain full column rank. Figure~\ref{fig:recoverability-validation}(c)--(d) shows that the layer-wise mean quotient margin is positive for both semantic bases at every evaluated layer. For Qwen, it ranges from 1.9119 to 44.0236 under the CelebA basis and from 1.6435 to 37.9883 under the CLIP-PC basis. For LLaVA, the corresponding ranges are 0.6620--2.9139 and 0.5740--2.3985. These positive margins indicate that the tested semantic variations remain distinguishable after removing the hidden-state responses associated with the measured nuisance variations. Since hidden-state scales differ across architectures, absolute margins are compared only within each model. To account for layer-dependent response scales, we additionally compute $\widehat\rho_{i,l}=\frac{\sigma_{\min}(\widehat S_{i,l})}{\sigma_{\min}(\widehat J_{c,i,l})}$ separately for each input before averaging. Across layers and semantic bases, the resulting mean $\widehat\rho_{i,l}$ ranges from 0.7952 to 0.9656 for Qwen and from 0.7010 to 0.9299 for LLaVA, showing that the weakest semantic response remains substantially preserved after nuisance projection.

These pointwise results provide empirical evidence for local semantic recoverability within the evaluated subspaces, as the selected semantics remain linearly decodable and non-degenerate with positive quotient margins after removing measured nuisance responses. They do not establish neighborhood-wide positivity, global or exact pixel-level invertibility, or recoverability under unmeasured nuisance factors.

\begin{table}[!t]
\centering
\caption{Quantitative comparison of image quality.}
\label{tab:main_results}
\setlength{\tabcolsep}{2.5pt}
\renewcommand{\arraystretch}{1.05}
\resizebox{\columnwidth}{!}{
\begin{tabular}{llcccccc}
\toprule
\multirow{2}{*}{\textbf{Dataset}} &
\multirow{2}{*}{\textbf{Method}} &
\multicolumn{3}{c}{\textbf{Qwen-VL}} &
\multicolumn{3}{c}{\textbf{LLaVA}} \\
\cmidrule(lr){3-5}
\cmidrule(lr){6-8}
& &
MSE $\downarrow$ &
PSNR $\uparrow$ &
SSIM $\uparrow$ &
MSE $\downarrow$ &
PSNR $\uparrow$ &
SSIM $\uparrow$ \\
\midrule

% ===================== CelebA =====================
\multirow{4}{*}{CelebA}
& SMI
& 0.2504 & 6.28 & 0.0195
& 0.2531 & 6.32 & 0.0200 \\
& CLIPInversion
& 0.1780 & 7.57 & 0.0139
& 0.1989 & 7.07 & 0.0096 \\
& INVERSE-TVM
& 0.0303 & 15.58 & 0.4685
& 0.0335 & 15.17 & 0.4761 \\
\rowcolor{gray!15}\cellcolor{white}
& Ours
& \textbf{0.0133} & \textbf{19.23} & \textbf{0.6256}
& \textbf{0.0135} & \textbf{19.24} & \textbf{0.6385} \\
\midrule

% ===================== Flowers-102 =====================
\multirow{4}{*}{Flowers-102}
& SMI
& 0.2129 & 6.88 & 0.0231
& 0.2041 & 7.09 & 0.0220 \\
& CLIPInversion
& 0.1834 & 7.43 & 0.0122
& 0.2014 & 7.02 & 0.0092 \\
& INVERSE-TVM
& 0.0465 & 13.68 & 0.3390
& 0.0489 & 13.43 & 0.3429 \\
\rowcolor{gray!15}\cellcolor{white}
& Ours
& \textbf{0.0222} & \textbf{17.08} & \textbf{0.4630}
& \textbf{0.0204} & \textbf{17.51} & \textbf{0.4879} \\
\midrule

% ===================== VQAv2 =====================
\multirow{4}{*}{VQAv2}
& SMI
& 0.2379 & 6.50 & 0.0138
& 0.2370 & 6.54 & 0.0140 \\
& CLIPInversion
& 0.1687 & 7.81 & 0.0130
& 0.1861 & 7.37 & 0.0099 \\
& INVERSE-TVM
& 0.0352 & 15.07 & 0.3990
& 0.0396 & 14.51 & 0.3872 \\
\rowcolor{gray!15}\cellcolor{white}
& Ours
& \textbf{0.0213} & \textbf{17.39} & \textbf{0.4797}
& \textbf{0.0249} & \textbf{16.61} & \textbf{0.4504} \\
\bottomrule
\end{tabular}
}
\end{table}

\begin{table}[t]
\centering
\caption{Text reconstruction performance (\%).}
\label{tab:text_reconstruction_results}
\scriptsize
\setlength{\tabcolsep}{1.8pt}
\renewcommand{\arraystretch}{1.0}
\resizebox{\columnwidth}{!}{
\begin{tabular}{@{}llcccccc@{}}
\toprule
\multirow{2}{*}{Dataset} &
\multirow{2}{*}{Method} &
\multicolumn{3}{c}{Qwen} &
\multicolumn{3}{c}{LLaVA} \\
\cmidrule(lr){3-5}
\cmidrule(lr){6-8}
& &
Token Acc. $\uparrow$ & EMR $\uparrow$ & BERT-F1 $\uparrow$ &
Token Acc. $\uparrow$ & EMR $\uparrow$ & BERT-F1 $\uparrow$ \\
\midrule

\multirow{3}{*}{FEVER}
& PIA
& 27.45 & 0.00 & 53.31
& 16.41 & 0.00 & 48.14 \\
& SIPIT
& 96.57 & 78.33 & 96.92
& 98.88 & 90.00 & 99.34 \\
\rowcolor{gray!15}\cellcolor{white}
& Ours
& \textbf{97.93} & \textbf{84.17} & \textbf{98.65}
& \textbf{99.56} & \textbf{94.17} & \textbf{99.63} \\
\midrule

\multirow{3}{*}{MS MARCO}
& PIA
& 28.13 & 4.17 & 58.72
& 17.98 & 0.83 & 52.79 \\
& SIPIT
& 95.98 & 84.17 & 96.97
& 99.20 & 93.33 & 99.21 \\
\rowcolor{gray!15}\cellcolor{white}
& Ours
& \textbf{99.76} & \textbf{98.33} & \textbf{99.81}
& \textbf{99.90} & \textbf{99.17} & \textbf{99.97} \\
\midrule

\multirow{3}{*}{VQAv2}
& PIA
& 28.26 & 0.00 & 53.21
& 19.11 & 0.00 & 49.02 \\
& SIPIT
& 83.24 & 38.33 & 83.54
& 98.48 & 93.33 & 99.55 \\
\rowcolor{gray!15}\cellcolor{white}
& Ours
& \textbf{99.89} & \textbf{99.17} & \textbf{99.95}
& \textbf{99.80} & \textbf{99.17} & \textbf{99.92} \\
\bottomrule
\end{tabular}
}
\end{table}

\begin{figure}[!t]
    \centering
    \includegraphics[width=0.98\columnwidth]{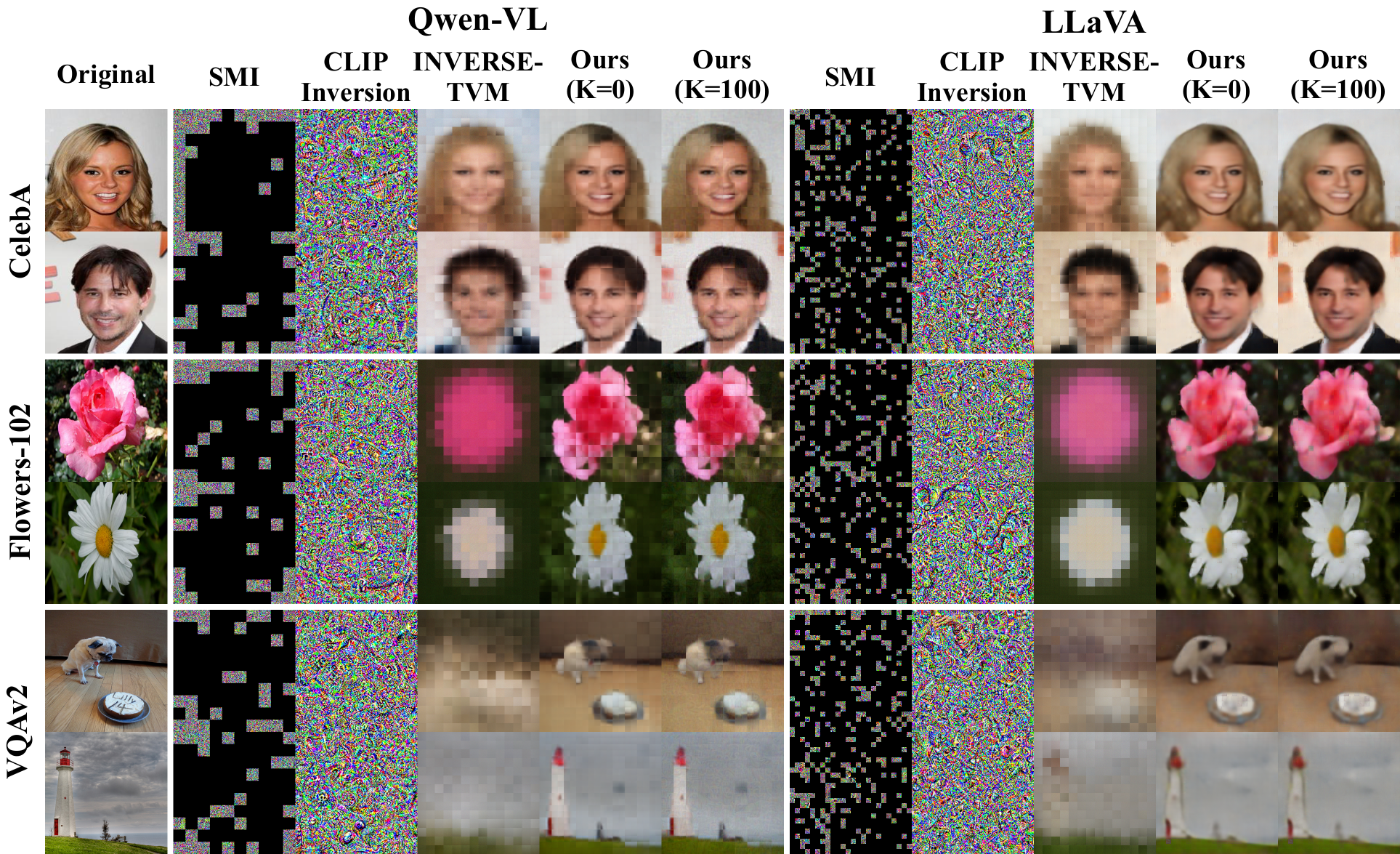}
    \caption{Visual comparison across different methods.}
    \label{fig:visual_reconstruction}
\end{figure}

\subsection{Effectiveness of Multimodal Reconstruction (RQ2)}
\label{sec:rq2-reconstruction}

\subsubsection{\textbf{Visual Reconstruction}}

To evaluate RASR across visual domains and LVLM architectures, we compare it with three baselines on CelebA, Flowers-102, and VQAv2, which cover faces, fine-grained flowers, and diverse real-world scenes, using Qwen and LLaVA. \textbf{RASR consistently achieves the strongest quantitative performance.} As shown in Table~\ref{tab:main_results}, it obtains the lowest MSE and highest PSNR and SSIM in all six dataset--model combinations. Compared with INVERSE-TVM, the strongest baseline, RASR reduces MSE by 50.5\% on average, improves PSNR by 3.27\,dB, and yields a 30.5\% average relative SSIM gain. Because aggregate metrics do not fully reflect input-specific content, we further examine representative reconstructions from all three datasets under both models. \textbf{RASR also produces more faithful reconstructions.} As shown in Fig.~\ref{fig:visual_reconstruction}, CLIPInversion recovers limited meaningful structure, while INVERSE-TVM preserves coarse content but exhibits structural and color distortions. In contrast, RASR better retains spatial layout, color distribution, and input-specific appearance, particularly after consistency-guided refinement. These results show that exposed hidden states preserve recoverable visual information across the evaluated domains and architectures.

\subsubsection{\textbf{Textual Reconstruction}}

To evaluate textual leakage across task types and LVLM architectures, we reconstruct inputs from FEVER, MS MARCO, and VQAv2 using the intercepted hidden states of Qwen and LLaVA. These datasets cover fact verification, passage retrieval, and visual question answering. \textbf{RASR achieves the strongest reconstruction accuracy.} As shown in Table~\ref{tab:text_reconstruction_results}, it performs best on all three metrics across all six dataset--model combinations. Across these settings, RASR averages 99.47\% Token Acc., 95.69\% EMR, and 99.66\% BERT-F1, exceeding SIPIT by 4.08, 16.12, and 3.73 percentage points, respectively. \textbf{RASR is also substantially faster.} Across the three datasets, it requires 0.193/0.267\,s/token for Qwen/LLaVA, compared with 33.731/42.192\,s/token for PIA and 60.553/160.997\,s/token for SIPIT. These results show that split-layer hidden states enable accurate and efficient recovery of both exact textual content and semantic information.

\begin{table}[t]
\centering
\caption{Reconstruction performance across partition layers on CelebA (visual) and FEVER (textual).}
\label{tab:all-layer-recovery}
\scriptsize
\setlength{\tabcolsep}{2.5pt}
\renewcommand{\arraystretch}{0.95}

\resizebox{\columnwidth}{!}{
\begin{tabular}{@{}llccccccc@{}}
\toprule
Model & Metric &
\multicolumn{7}{c}{Reconstruction at different layers} \\
\cmidrule(lr){3-9}

\multirow{5}{*}{Qwen}
& Layer
& 0 & 6 & 11 & 17 & 23 & 29 & 35 \\
\cmidrule(lr){2-9}

& PSNR $\uparrow$
& 21.69 & 21.89 & 20.46 & 19.28
& 19.23 & 18.44 & 17.66 \\

& Sensitive Acc. $\uparrow$
& 88.36\% & 89.04\% & 87.97\%
& 86.83\% & 87.22\% & 86.23\% & 84.41\% \\
\cmidrule(lr){2-9}

& Token Acc. $\uparrow$
& 100.00\% & 99.57\% & 99.21\% & 97.64\%
& 97.93\% & 96.07\% & 91.35\% \\

& BERT-F1 $\uparrow$
& 100.00\% & 99.74\% & 99.63\% & 97.93\%
& 98.65\% & 96.57\% & 93.24\% \\
\midrule

\multirow{5}{*}{LLaVA}
& Layer
& 0 & 5 & 10 & 15 & 20 & 25 & 31 \\
\cmidrule(lr){2-9}

& PSNR $\uparrow$
& \textbf{21.21} & 21.06 & 20.34 & 19.73
& 19.24 & 19.02 & 18.58 \\

& Sensitive Acc. $\uparrow$
& 89.54\% & 89.49\% & 89.12\%
& 89.37\% & 89.26\% & 89.31\% & 89.11\% \\
\cmidrule(lr){2-9}

& Token Acc. $\uparrow$
& 100.00\% & 99.88\% & 100.00\%
& 99.88\% & 99.56\% & 98.69\% & 96.82\% \\

& BERT-F1 $\uparrow$
& 100.00\% & 99.95\% & 100.00\%
& 99.92\% & 99.63\% & 98.44\% & 97.06\% \\
\bottomrule
\end{tabular}
}
\end{table}

\subsubsection{\textbf{Partition Depth}}

Since collaborative LVLM inference may use different partition depths, we evaluate RASR at seven layers per model on CelebA and FEVER for visual and textual recovery, respectively. On CelebA, we use the pretrained Anycost GANs attribute classifier~\cite{Lin_2021_CVPR} to measure ten-attribute macro accuracy (Sensitive Acc.), with 94.22\% on the original images as the reference. \textbf{Visual leakage persists across partition depths.} As shown in Table~\ref{tab:all-layer-recovery}, although image reconstruction quality generally declines with depth, Sensitive Acc. remains above 84\% at the deepest evaluated layer for both models, indicating that privacy-relevant visual attributes remain recoverable. \textbf{Textual leakage also persists across partition depths.} Token Acc. exhibits an overall decline with depth but remains above 91\% at the deepest evaluated layer for both models, showing that substantial input-text information persists even in deep-layer hidden states.

\begin{table}[t]
\centering
\caption{Ablation of the visual reconstructor architecture on CelebA using Qwen at layer 23.}
\label{tab:reconstructor_ablation}
\scriptsize
\setlength{\tabcolsep}{2.5pt}
\renewcommand{\arraystretch}{1.05}
\resizebox{\columnwidth}{!}{
\begin{tabular}{@{}lcccc@{}}
\toprule
Variant &
Effective Params. ($\Delta$) &
MSE $\downarrow$ &
PSNR $\uparrow$ &
SSIM $\uparrow$ \\
\midrule
w/o Hierarchical Mapping
& 696.85M $(+0.20\%)$
& 0.0154
& 18.56
& 0.6039 \\

w/o Inverse Patch Merger
& 690.15M $(-0.76\%)$
& 0.0141
& 18.97
& 0.6178 \\

\textbf{Full}
& 695.47M $(0.00\%)$
& \textbf{0.0133}
& \textbf{19.23}
& \textbf{0.6256} \\
\bottomrule
\end{tabular}
}
\end{table}

\begin{figure}[!t]
    \centering
    \includegraphics[width=0.98\columnwidth]{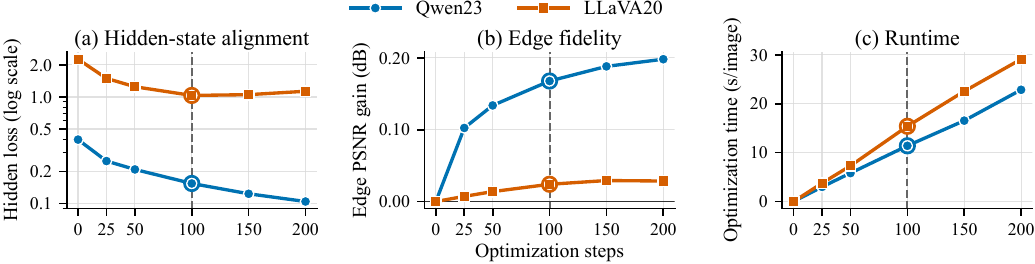}
    \caption{Effect of the visual refinement budget $K$.}
    \label{fig:visual_refinement_steps}
\end{figure}

\begin{table}[t]
\centering
\caption{Effect of reconstructed image guidance on textual refinement on VQAv2 (\%).}
\label{tab:reconstructed_image_ablation}
\scriptsize
\setlength{\tabcolsep}{2.2pt}
\renewcommand{\arraystretch}{1.0}
\resizebox{\columnwidth}{!}{
\begin{tabular}{@{}lcccccc@{}}
\toprule
\multirow{2}{*}{Setting} &
\multicolumn{3}{c}{Qwen} &
\multicolumn{3}{c}{LLaVA} \\
\cmidrule(lr){2-4}
\cmidrule(lr){5-7}
& Token Acc. & EMR & BERT-F1
& Token Acc. & EMR & BERT-F1 \\
\midrule

w/o reconstructed image
& 99.78 & 98.33 & 99.89
& 99.70 & 98.33 & 99.89 \\

with reconstructed image
& \textbf{99.89} & \textbf{99.17} & \textbf{99.95}
& \textbf{99.80} & \textbf{99.17} & \textbf{99.92} \\

\bottomrule
\end{tabular}
}
\end{table}

\begin{table}[!t]
\centering
\caption{Ablation study of the confidence threshold $\tau$.}
\label{tab:text_refinement_ablation}
\setlength{\tabcolsep}{3.5pt}
\renewcommand{\arraystretch}{1.06}
\resizebox{\columnwidth}{!}{
\begin{tabular}{lcccc}
\toprule
\textbf{Setting} &
\textbf{Token Acc.} $\uparrow$ &
\textbf{EMR} $\uparrow$ &
\textbf{BERT-F1} $\uparrow$ &
\textbf{Time (ms/token)} $\downarrow$ \\
\midrule
w/o refinement
& 96.16 & 73.33 & 96.92 & 11.25 \\

$\tau=0.75$
& 99.44 & 95.28 & 99.63 & 201.38 \\

$\tau=0.80$
& 99.44 & 95.28 & 99.63 & 214.70 \\

$\tau=0.85$
& \textbf{99.47} & 95.69 & \textbf{99.66} & 229.74 \\

$\tau=0.90$
& 99.45 & 95.69 & \textbf{99.66} & 268.91 \\

$\tau=0.95$
& 99.45 & \textbf{95.83} & 99.64 & 315.32 \\
\bottomrule
\end{tabular}
}
\end{table}

\subsection{Ablation and Parameter Analysis (RQ2)}

\subsubsection{\textbf{Visual reconstructor ablation}} We assess the contributions of hierarchical mapping and the inverse patch merger by separately removing each component on CelebA using Qwen at layer 23. As shown in Table~\ref{tab:reconstructor_ablation}, the full reconstructor performs best across all three metrics. Its MSE rises from 0.0133 to 0.0154 without hierarchical mapping and to 0.0141 without the inverse patch merger. Notably, the former variant uses 0.20\% more effective parameters, showing that the benefit of hierarchical mapping is not attributable to parameter count.

\subsubsection{\textbf{Visual refinement steps}} We vary $K$ on CelebA at the default partition layers to assess effectiveness and runtime. As shown in Fig.~\ref{fig:visual_refinement_steps}, relative to $K=0$, $K=100$ reduces hidden-state MSE by 61.25\% for Qwen and 54.30\% for LLaVA, while improving Edge PSNR by only 0.168 and 0.024\,dB, respectively. This limited gain may reflect the approximately many-to-one mapping of fine-grained image details into split-layer hidden states. The constraint narrows the solution space around the initial reconstruction but cannot uniquely recover all edge details. Beyond $K=100$, Edge PSNR changes little, LLaVA's hidden-state MSE rises, and runtime increases. Therefore, we set $K=100$ by default.

\subsubsection{\textbf{Visual guidance for textual refinement}} We examine whether refined-image guidance improves text recovery on VQAv2, the only text benchmark with paired images. As shown in Table~\ref{tab:reconstructed_image_ablation}, it improves Token Acc. by 0.11/0.10 percentage points for Qwen/LLaVA and EMR by 0.84 percentage points for both, with slight BERT-F1 gains. Although modest given the high initial accuracy, these consistent gains indicate that refined visual context provides complementary information for text recovery.

\subsubsection{\textbf{Textual refinement threshold}} We sweep $\tau$ across all six dataset--model combinations to balance refinement coverage and computational cost. As shown in Table~\ref{tab:text_refinement_ablation}, $\tau=0.85$ improves all three recovery metrics over no refinement, reaching 99.47\% Token Acc. and 95.69\% EMR at 229.74\,ms/token. Higher thresholds increase runtime to 268.91--315.32\,ms/token without yielding consistent accuracy gains. Therefore, we use $\tau=0.85$ as the default threshold.
\section{Conclusion and Future Work}

This work addresses the underexplored risk of multimodal privacy leakage in collaborative LVLM inference. We first show that, under local regularity conditions and a positive semantic--nuisance separation margin, privacy-relevant visual semantics remain locally identifiable and stably recoverable from intermediate hidden states. Building on this analysis, we propose RASR, a coarse-to-fine attack that combines symmetric reconstruction with hidden-state-guided refinement to recover visual and textual inputs. Experiments on two LVLMs and five datasets show that RASR consistently outperforms existing attacks in both visual and textual reconstruction. By demonstrating the recoverability of both modalities, these results underscore the need for lightweight representation-level defenses that suppress inversion-sensitive information while preserving inference utility, as well as selective privacy-preserving computation that protects only high-risk hidden-state components or layers to reduce computational and communication overhead.

\bibliographystyle{IEEEtran} 
\bibliography{main}

\end{document}